\documentclass[12pt]{article}

\usepackage{amsmath}
\usepackage{amssymb}
\usepackage{amsfonts}
\usepackage{latexsym}
\usepackage{color}

\catcode `\@=11 \@addtoreset{equation}{section}
\def\theequation{\arabic{section}.\arabic{equation}}
\catcode `\@=12

\newcommand{\be}{\begin{equation}}
	\newcommand{\en}{\end{equation}}
\newcommand{\bea}{\begin{eqnarray}}
	\newcommand{\ena}{\end{eqnarray}}
\newcommand{\beano}{\begin{eqnarray*}}
	\newcommand{\enano}{\end{eqnarray*}}
\newcommand{\bee}{\begin{enumerate}}
	\newcommand{\ene}{\end{enumerate}}

\newcommand{\B}{{\mathfrak B}}

\newcommand{\mc}{\mathcal}

\newcommand{\Sc}{{\cal S}}

\newcommand{\F}{{\cal F}}

\newcommand{\1}{1 \!\! 1}

\newcommand{\Hil}{\mc H}

\renewcommand{\l}{\langle}
\renewcommand{\r}{\rangle}
\newcommand{\pint}{\l\cdot,\cdot\r}
\newcommand{\pin}[2]{\l#1 , #2\r}
\newcommand{\nor}{\|\cdot\|}

\newtheorem{thm}{Theorem}

\newtheorem{lemma}[thm]{Lemma}
\newtheorem{prop}[thm]{Proposition}
\newtheorem{defn}[thm]{Definition}

\newenvironment{proof}{\noindent {\bf Proof --}}{\hfill$\square$ \vspace{3mm}\endtrivlist}

\catcode `\@=11 \@addtoreset{equation}{section}
\catcode `\@=12

\begin{document}

\thispagestyle{empty}

\vspace*{2cm}

\begin{center}
{\Large \bf Heisenberg dynamics for non self-adjoint Hamiltonians: symmetries and derivations}   \vspace{2cm}\\

{\large F. Bagarello}\\
Dipartimento di Ingegneria,
Universit\`a di Palermo,\\ I-90128  Palermo, Italy\\
and I.N.F.N., Sezione di Catania\\
e-mail: fabio.bagarello@unipa.it\\

\end{center}

\vspace*{1cm}

\begin{abstract}

In some recent literature the role of non self-adjoint Hamiltonians, $H\neq H^\dagger$, is often considered in connection with gain-loss systems. The dynamics for these systems is, most of the times, given in terms of a Schr\"odinger equation. In this paper we rather focus on the Heisenberg-like picture of quantum mechanics, stressing the (few) similarities and the (many) differences with respected to the standard Heisenberg picture for systems driven by self-adjoint Hamiltonians. In particular, the role of the symmetries, *-derivations and integrals of motion is discussed.

\end{abstract}

\vspace{2cm}


\vfill


\newpage

\section{Introduction}

In the past years non-Hermitian Hamiltonians have played, and still it plays, a relevant role in quantum mechanics. This is because self-adjointess of the observables is a sufficient, but not necessary, condition for the reality of their eigenvalues, which is usually considered as an essential aspect of any realistic quantum model. This is clearly shown, for instance, by the cubic oscillator considered in \cite{ben1}, whose Hamiltonian has real eigenvalues, even in presence of a purely imaginary potential. Since then, a lot of work was done to achieve a deeper comprehension of this kind of operators, both for their physical relevance, and for their interesting mathematical properties. Some relevant monographs and edited books are  \cite{benbook}-\cite{bagspringer}.

Most of the times, even in presence of an Hamiltonian $H\neq H^\dagger$, which for simplicity we will always suppose here not to depend explicitly on time, the dynamics is thought to be deduced from a Schr\"odinger equation, $i\dot \Psi(t)=H\Psi(t)$, as for self-adjoint Hamiltonians. This is clearly an implicit assumption, since there is no mathematical reason, a priori, to prefer the previous equation rather than the analogous  Schr\"odinger equation for $H^\dagger$, $i\dot \Phi(t)=H^\dagger\Phi(t)$, except that for the physical implications of one choice or the other. However, any of these choices produce an Heisenberg-like dynamics whose properties are often not considered in detail, and which have very special features. For instance, the time evolution of the product of two observables $X$ and $Y$, $(XY)(t)$, is in general different from $X(t)Y(t)$. Also, the operators which commute with $H$, or with $H^\dagger$, are not constants of motion and, in this context, it is not entirely clear what a (mathematical? physical?) symmetry is. And yet, the generator of the dynamics is neither $H$ nor $H^\dagger$, and the derivation connected to this time evolution is not an algebraic derivation in the usual sense, since it does not satisfy the Leibnitz rule for the derivative of a product. On the other hand, intertwining operators appear to play a significant role, which should be fully understood. 

These are some of the main aspects we are going to consider in a sistematic way in the rest of the paper. In particular, to make the treatment of these problems significantly simpler, in this paper we will focus on finite dimensional Hilbert spaces since this guarantees the fact that we deal only with bounded operators. This is useful to avoid all the technicalities connected with the fact that, if the Hilbert space has dimension infinity, unbounded operators usually play a relevant role in quantum mechanics.

The paper is organized as follows:

In the next section we introduce the notation and list some preliminary results.

Then, in Section \ref{sect3}, we propose our definitions of $\gamma$-dynamics and $\gamma$-derivation, and we study their properties and their relations, in a rather general settings.

In Section \ref{sect4} we introduce a metric operator, and we show how this operator can be used in the analysis of what we call $\gamma$-symmetries and of their time evolution.

More adjoint maps are considered in Section \ref{sect5}, while Section \ref{sectconcl} contains our conclusions, and plans for the future. The Appendix contains a simple example relevant for what discussed in the main part of the paper.

  \section{The mathematical settings}\label{sect2}

  Let $\Hil$ be an Hilbert space, with $dim(\Hil)=N<\infty$, with scalar product $\pint$ and associated norm $\nor$, where $\|f\|=\sqrt{\pin{f}{f}}$, $\forall f\in\Hil$. As usual we have $\left<f,g\right>=\sum_{k=1}^N\overline{f_k}\,g_k$, $f,g\in\Hil$. The scalar product is also linked to the  conjugation $\dagger$ defined as $\pin{X^\dagger f}{g}=\pin{f}{Xg}$, $\forall f,g\in\Hil$. Here $X$ is an operator on $\Hil$ which, in our particular case, is a matrix of dimension $N\times N$. More in general, also in view of future extensions to infinite-dimensional Hilbert spaces, we could say that $X\in\B(\Hil)$, the $C^*$-algebra of all bounded operators on $\Hil$, \cite{br,sewell}. In doing this the dimensionality of $\Hil$ could also be infinite.

  The main ingredient of our analysis is an operator (i.e. a matrix) $H$, acting on $\Hil=\mathbb{C}^{N}$, with $H\neq H^\dagger$. We assume for simplicity (but this does not affect our conclusions) that $H$ has exactly $N$ distinct eigenvalues $E_n$, $n=0,1,2,\ldots,N$. Here, the adjoint $H^\dagger$ of $H$ is the usual one, i.e. the complex conjugate of the transpose of the matrix $H$. 
  
  Other than being all different, we will also assume that $E_n\in\mathbb{R}$, $n=1,2,,\ldots,N$. This is more important since, in this way, we can check that $H$ and $H^\dagger$ are isospectral and, as such, they admit an intertwining operator $X$ such that $XH=H^\dagger X$. This isospectrality is lost if even only one eigenvalue of $H$ is complex. This situation has been considered, for instance, in \cite{bagAoP1,bagAoP2}, to which we refer for more details.

 The $N$ distinct real eigenvalues of $H$ correspond to $N$ distinct eigenvectors $\varphi_k$, $k=1,2,\ldots,N$:
  \be
  H\varphi_k=E_k\varphi_k.
  \label{21}\en
  The set $\F_\varphi=\{\varphi_k,\,k=1,2,\ldots,N\}$ is a basis for $\Hil$, since the eigenvalues are all different. Then an unique biorthogonal basis of $\Hil$, $\F_\Psi=\{\Psi_k,\,k=1,2,\ldots,N\}$, surely exists, \cite{you,chri}: $\left<\varphi_k,\Psi_l\right>=\delta_{k,l}$, for all $k, l$. Moreover, for all $f\in\Hil$, we can write
  $f=\sum_{k=1}^N\left<\varphi_k,f\right>\Psi_k = \sum_{k=1}^N\left<\Psi_k,f\right>\varphi_k$. In \cite{bagAoP1} we have shown that the vectors in $\F_\Psi$ are eigenstates of $H^\dagger$ with eigenvalues $E_k$:
  \be
  H^\dagger\Psi_k=E_k\Psi_k,
  \label{22}\en
  $k=0,1,2,\ldots,N$. 
  
  Using the bra-ket notation we can write $\sum_{k=0}^N|\varphi_k\left>\right<\Psi_k|=\sum_{k=0}^N|\Psi_k\left>\right<\varphi_k|=\1$, where, for all $f,g,h\in\Hil$, we define $(|f\left>\right<g|)h:=\left<g,h\right>f$, and $\1$ is the identity operator on $\Hil$. We introduce the operators $S_\varphi=\sum_{k=0}^N|\varphi_k\left>\right<\varphi_k|$ and $S_\Psi=\sum_{k=0}^N|\Psi_k\left>\right<\Psi_k|$, as in \cite{bagbook}. These are bounded positive, self-adjoint, invertible operators, one the inverse of the other: $S_\Psi=S_\varphi^{-1}$. They are often called {\em metric operators}, since they can be used to define new scalar products in $\Hil$. This will be relevant later, in Section \ref{sect5}. Moreover:
  \be
  S_\varphi\Psi_n=\varphi_n,\quad S_\Psi\varphi_n=\Psi_n,\quad \mbox{ as well as }\quad S_\Psi H=H^\dagger S_\Psi,\quad S_\varphi H^\dagger= HS_\varphi.
  \label{23}\en
 Many other details are discussed in \cite{bagAoP1}. Here, what is more relevant for us, is to introduce the Heisenberg-like dynamics we will consider next, and discuss its properties\footnote{We observe that other possible Heisenberg-like dynamics can also be introduced, see \cite{bagAoP1,bagAoP2}, which are connected to some other scalar products we could also introduce on $\Hil$, as in Section \ref{sect5}.}. To fix the ideas, we consider the   Schr\"odinger equation $i\dot \Psi(t)=H\Psi(t)$ as the starting point of our analysis. Hence we introduce a map $\gamma^t$ on $\B(\Hil)$ as follows:
 \be
 \pin{\varphi(t)}{X\psi(t)}=\pin{\varphi_0}{\gamma^t(X)\psi_0},
 \label{24}\en
 where $\varphi(t)$ and $\psi(t)$ both obey the Schr\"odinger equation, with initial values  $\varphi(0)=\varphi_0$ and $\psi(0)=\psi_0$. Then we have
 \be
 \gamma^t(X)=e^{iH^\dagger t}Xe^{-iHt}.
 \label{25}\en
 This is what we will call {\em $\gamma$-dynamics}. It is clear that, in the present context, $\gamma^t(X)$ is well defined for all $X\in\B(\Hil)$. In fact, $\gamma^t(X)$ is simply the product of three $N\times N$ matrices or, more in general, of three bounded operators. It is also clear that, if $H=H^\dagger$, then $\gamma^t(X)=e^{iHt}Xe^{-iHt}$, which is exactly the standard Heisenberg dynamics for a self-adjoint Hamiltonian. We end this section by listing a series of well-known and useful properties of this latter, which we will try to reconsider for the non-Hermitian case in what follows. 
 
 For that, let us consider an operator $H_0=H_0^\dagger$, independent on time and such that $H_0\in\B(\Hil)$, and let us call
 \be
 \alpha_0^t(X)=e^{iH_0t}Xe^{-iH_0t},
 \label{26}\en
 $X\in\B(\Hil)$. Then we have:
 \begin{enumerate}
 	\item $\alpha_0^t(\1)=\1$, $\forall t\in\mathbb{R}$, where $\1$ is the identity operator on $\Hil$. Hence $\alpha_0^t$ {\em preserves } the identity operator.
 	\item $\alpha_0^t(XY)=\alpha_0^t(X)\alpha_0^t(Y)$, $\forall X,Y\in\B(\Hil)$: $\alpha_0^t$ is an authomorphism of $\B(\Hil)$.
 	\item $\alpha_0^t$ preserves the adjoint: $\alpha_0^t(X^\dagger)= (\alpha_0^t(X))^\dagger$, $\forall X\in\B(\Hil)$.
 	\item $\alpha_0^t$ is norm continuous: if $\{X_n\}$ converges in the norm of $\B(\Hil)$ to $X$, $\|X_n-X\|\rightarrow 0$ for $n\rightarrow\infty$, then $\|\alpha_0^t(X_n)-\alpha_0^t(X)\|\rightarrow 0$ in the same limit.
 	\item if $Z\in\B(\Hil)$ commutes with $H_0$, $[H_0,Z]=0$, then $\alpha_0^t(Z)=Z$,  $\forall t\in\mathbb{R}$: all the operators commuting with $H_0$ are {\em constant of motion}.
 	\item if we introduce $\delta_0(X)=\lim_{t,0}\frac{\alpha_0^t(X)-X}{t}$, the limit to be understood in the norm of $\B(\Hil)$, then $\delta_0$ is a *-derivation: (i) $\delta_0(X^\dagger)=(\delta_0(X))^\dagger$, and (ii) $\delta_0(XY)=X\delta_0(Y)+\delta_0(X)Y$, $\forall X,Y\in\B(\Hil)$.
 	\item The series $\sum_{k=0}^{\infty}\frac{t^k\delta_0^k(X)}{k!}$ is norm convergent to $\alpha_0^t(X)$ for all $X\in\B(\Hil)$. Here $\delta_0^k(X)$ is defined recursively as follows: $\delta_0^0(X)=X$, and $\delta_0^k(X)=\delta_0(\delta_0^{k-1}(X))$, $k\geq1$.
 \end{enumerate}

 \section{The $\gamma$-dynamics}\label{sect3}
 
 In this section we will discuss which of the previous properties of $\alpha_0^t$ are still true also for $\gamma^t$, and which are not. In our analysis we will try first to be rather generic, and then we will introduce and use an operator $S$ which intertwines between $H$ and $H^\dagger$ as in (\ref{23}). This {\em two-steps procedure} can be useful when extending our results to an infinite-dimensional Hilbert space, with unbounded observables, since in this case the existence and the properties of $\Sc$ are not as easy to deduce.

The first obvious result is that $\gamma^t(\1)=e^{iH^\dagger t}\,e^{-iHt}\neq\1$, since $H\neq H^\dagger$, as we will assume always, expect when stated. In particular, $\gamma^t(\1)$ depends explicitly on time, and obyes the differential equation $\frac{d\gamma^t(\1)}{dt}=\gamma^t(H^\dagger-H)$, with $\gamma^0(\1)=\1$. This will have interesting consequences. The first of these consequences is that  $\gamma^t(\1)$ encodes {\bf all} the differences between $\gamma^t$ and $\alpha_H^t$, where, in analogy with (\ref{26}), we put $\alpha_H^t(X)=e^{iHt}Xe^{-iHt}$. Indeed we have,
\be
\gamma^t(X)=\gamma^t(\1)\,\alpha_H^t(X)=(\alpha_H^t(X^\dagger))^\dagger\gamma^t(\1),
\label{31}\en
 $\forall X\in\B(\Hil)$. These identities are easily deduced. It is useful to stress that $\alpha_H^t$ has not all the same properties as $\alpha_0^t$ in the previous section. In particular, since $H\neq H^\dagger$, we observe that $\alpha_H^t(X^\dagger)\neq (\alpha_H^t(X))^\dagger$. However, we still have that  $\alpha_H^t(XY)=\alpha_H^t(X)\alpha_H^t(Y)$, $\forall X,Y\in\B(\Hil)$, and $\alpha_H^t(\1)=\1$.
 
 The following result holds:
 \begin{lemma}\label{lemma1}
 	The following statements are equivalent: 1) $\gamma^t(XY)=\gamma^t(X)\gamma^t(Y)$, $\forall X,Y\in\B(\Hil)$; and 2) $\gamma^t(\1)=\1$.
 \end{lemma}
 \begin{proof}
 	Using (\ref{31}) we can check that $\gamma^t(XY)=\gamma^t(X)\gamma^t(Y)$ if and only if $$\gamma^t(\1)\alpha_H^t(X)\left(\1-\gamma^t(\1)\right)\alpha_H^t(Y)=0.$$ Since this equality must hold for all $X,Y\in\B(\Hil)$, it must hold in particular if $X=Y=\1$. Hence, since $\alpha_H^t(\1)=\1$,  we must have  $\gamma^t(\1)\left(\1-\gamma^t(\1)\right)=0$, which implies that $\gamma^t(\1)=\1$ due to the invertibility of $\gamma^t(\1)$.
 	The opposite implication is obvious.
 	
 \end{proof}
Of course, in view of what we have seen before, $\gamma^t(\1)=\1$ only if $H=H^\dagger$. Hence this lemma is saying that $\gamma^t$ is an automorphism of $\B(\Hil)$ only for self-adjoint $H$. Stated differently, properties 1. and 2. of $\alpha_0^t$ are both lost by $\gamma^t$. On the other hand, properties 3. and 4. of $\alpha_0^t$ can also be deduced for $\gamma^t$. In particular, the norm continuity of $\gamma^t$ is a consequence of the fact that $H\in\B(\Hil)$. Indeed in this case, since $\|H\|<\infty$,  the series $\sum_{k=0}^\infty \frac{1}{k!}\,(iHt)^k$ is norm-convergent for all $t\in\mathbb{R}$ and this allows us to conclude that
$$
\|e^{\pm iHt}\|\leq e^{|t|\|H\|}, \qquad \|e^{\pm iH^\dagger t}\|\leq e^{|t|\|H\|},
$$
 $\forall t\in\mathbb{R}$. Hence, if $\|X_n-X\|\rightarrow 0$ for $n\rightarrow\infty$, then 
 $$
 \left\|\gamma^t(X_n)-\gamma^t(X)\right\|= \left\|e^{iH^\dagger t}(X_n-X)e^{-iHt}\right\|\leq $$
 $$\leq\left\|e^{iH^\dagger t}\right\|\left\|X_n-X\right\|\left\|e^{-iHt}\right\| \leq e^{2|t|\|H\|}\left\|X_n-X\right\|\rightarrow0,
 $$
 for $n\rightarrow\infty$, as we had to prove.

 Pushing forward our parallel between $\alpha_0^t$ and $\gamma^t$, we now define a map $\delta_\gamma:\B(\Hil)\rightarrow\B(\Hil)$ as follows:
 \be
 \delta_\gamma(X)=\|.\|-\lim_{t,0}\frac{\gamma^t(X)-X}{t}=i\left(H^\dagger X-XH\right),
 \label{32}\en
 $X\in\B(\Hil)$, where the last result follows from the continuity of $H$, $H^\dagger$, and of their exponentials. Notice that $ \delta_\gamma(X)\in\B(\Hil)$. This is what we call {\em $\gamma$-derivation}. In particular, extending a  standard definition in algebraic quantum theory, \cite{br,br2}, we say that $\delta_\gamma$ is {\em inner}, in the sense that we have an operator, $H$, which together with $H^\dagger$, {\em represents} the map $\delta_\gamma$ as a sort of generalized commutator.
 
 It is now interesting to discuss the main properties of $\delta_\gamma$, and its connections with $\delta_0$ and with $\gamma^t$.
 
 We first observe that $\delta_\gamma(X^\dagger)=(\delta_\gamma(X))^\dagger$, $\forall X\in\B(\Hil)$. It is also possible to check that $\delta_\gamma$ is norm-continuous. Indeed we have $\|\delta_\gamma(X_n)-\delta_\gamma(X)\|\leq2\|H\|\|X_n-X\|\rightarrow0$, for any sequence $\{X_n\}$ norm-convergent to $X$. Moreover $\delta_\gamma\in\B(\B(\Hil))$, or, stated explicitly, $\delta_\gamma$ is a linear operator on the set $\B(\Hil)$. An interesting result is given by the following proposition, where we introduce, as for $\delta_0$, $\delta_\gamma^0(X)=X$, and $\delta_\gamma^k(X)=\delta_\gamma(\delta_\gamma^{k-1}(X))$, $k\geq1$.
 \begin{prop}\label{prop2}
 	The series $\sum_{k=0}^{\infty}\frac{t^k\delta_\gamma^k(X)}{k!}$ is norm convergent to $\gamma^t(X)$, for all $X\in\B(\Hil)$. 
 \end{prop} 

\begin{proof}
	First of all, we prove that the series is norm convergent. For that, it is sufficient to observe that $\|\delta_\gamma^k(X)\|\leq (2\|H\|)^k\|X\|$, for all $k\geq0$. This is clearly true for $k=0$. Suppose now that this inequality holds for some $k$. Hence we have
	$$
\|\delta_\gamma^{k+1}(X)\|	=\|\delta_\gamma(\delta_\gamma^{k}(X))\|\leq 2\|H\| \|\delta_\gamma^{k}(X)\|\leq (2\|H\|)^{k+1}\|X\|,
	$$
	because of our induction hypothesis, as we had to check. Then we have
	$$
	\left\|\sum_{k=0}^{\infty}\frac{t^k\delta_\gamma^k(X)}{k!}\right\|\leq e^{2|t|\|H\|}\|X\|.
	$$
	Notice that this is the same bound we have previously found when checking the continuity of $\gamma^t$. Once we have seen that the series converge, we want to prove that it converges to $\gamma^t(X)$, for all $X\in\B(\Hil)$ and for all $t\in\mathbb{R}$. 
	
	We define $\beta^t(X)=\gamma^t(X)-\sum_{k=0}^{\infty}\frac{t^k\delta_\gamma^k(X)}{k!}$, which is of course well defined in $\B(\Hil)$, also in view of what we have just proven. We have $\beta^0(X)=X-X=0$. Next we observe that
	\be
	\frac{d}{dt}\gamma^t(X)=\gamma^t(\delta_\gamma(X))=\delta_\gamma(\gamma^t(X)),
	\label{33}\en
	while, with standard computations,
	$$
	\frac{d}{dt}\sum_{k=0}^{\infty}\frac{t^k\delta_\gamma^k(X)}{k!}=\delta_\gamma\left(\sum_{k=0}^{\infty}\frac{t^k\delta_\gamma^k(X)}{k!}\right),
	$$
	where we have used the linearity of $\delta_\gamma$ and the fact that it is continuous. Therefore we have, using once more the linearity of $\delta_\gamma$, 
$$
\frac{d}{dt}\beta^t(X)=\delta_\gamma(\gamma^t(X))-\delta_\gamma\left(\sum_{k=0}^{\infty}\frac{t^k\delta_\gamma^k(X)}{k!}\right)=\delta_\gamma(\beta^t(X))=i\left(H^\dagger\beta^t(X)-\beta^t(X) H \right).
$$	
	It follows that $\beta^t(X)=0$ is a solution of this differential equation, with the correct initial condition. Because of Theorem 3.4 of \cite{bagmor}, we conclude that this is the only solution. Hence our claim follows.

\end{proof}

 So far we have shown that $\delta_\gamma$ is not so different from $\delta_0$. In particular, Proposition \ref{prop2} is the new form of property 7. for $\alpha_0^t$. However, this is not the end of the story:  let us now focus on the differences between $\delta_0$ and $\delta_\gamma$.
 
 The first evident difference is that, while $\delta_0(\1)=0$, $\delta_\gamma(\1)=i(H^\dagger-H)$, which is zero only if $H=H^\dagger$. In this case, and only in this case, $\delta_\gamma$ is a *-derivation. More in detail, we have the following proposition:
 \begin{prop}\label{prop3}
 	The following statements are equivalent: 1) $\delta_\gamma$ is a *-derivation; 2) $\delta_\gamma(\1)=0$; 3) $H=H^\dagger$; 4) $\gamma^t(\1)=\1$; 5)  $\gamma^t(XY)=\gamma^t(X)\gamma^t(Y)$, $\forall X,Y\in\B(\Hil)$.
 \end{prop}
\begin{proof}
	We have just seen that  $\delta_\gamma(X^\dagger)= (\delta_\gamma(X))^\dagger$. Let us introduce $\Delta_\gamma(XY)=\delta_\gamma(XY)-X\delta_\gamma(Y)-\delta_\gamma(X)Y$. $\delta_\gamma$ is a *-derivation if and only if $\Delta_\gamma(XY)=0$ for all $X,Y\in\B(\Hil)$. Using (\ref{32}) we find that $\Delta_\gamma(XY)=-X\delta_\gamma(\1)Y$. Since this should be zero for all $X$ and $Y$, it must be zero, in particular, for those $X$ and $Y$ which admit inverse. Hence $\delta_\gamma(\1)=0$. Vice-versa, the same equality shows that, if $\delta_\gamma(\1)=0$, then $\Delta_\gamma(XY)=0$ for all $X,Y\in\B(\Hil)$. Then, $\delta_\gamma$ is a *-derivation. This proves the equivalence between 1) and 2). We have just observed that   $\delta_\gamma(\1)=0$ if and only if $H=H^\dagger$. And we have also observed that this latter is equivalent to having $\gamma^t(\1)=\1$. The equivalence between 4) and 5) was proved in Lemma \ref{lemma1}.
	
\end{proof}
 
 The main content of this proposition is that $\gamma^t$ cannot be an authomorphism of $\B(\Hil)$ if any of the above properties (and therefore all) is violated. Another interesting, and serious, consequence of this proposition is discussed at length in the following section, and it is related to the possibility of deducing the dynamics of $\Sc$ using the $\gamma^t$-dynamics.

 \subsection{Consequences of Proposition \ref{prop3}}\label{sectnonauto}
 
 In quantum mechanics, when looking for the time evolution of a given system $\Sc$ driven by a time-independent self-adjoint Hamiltonian $H_0$, a possible approach is to use the Heisenberg representation as in (\ref{26}). As we have seen, $\alpha_0^t$ is a *-authomorphism.
 
 Now, if we try to repeat the same approach when $H_0$ is replaced by $H\neq H^\dagger$, and consequently $\alpha_0^t$ is replaced by $\gamma^t$ in (\ref{25}), serious difficulties arise as a consequence of Proposition \ref{prop3}. Indeed, let us suppose that $\Sc$ is completely described, for simplicity, by just two observables $A$ and $B$. Here, in view of what we have discussed so far, we will not insist on the fact that these two observables are self-adjoint or not. Using (\ref{33}) we have
\be
 \begin{cases}
 	& \frac{d}{dt}\gamma^t(A)=\gamma^t(\delta_\gamma(A)) \\
 	&\frac{d}{dt}\gamma^t(B)=\gamma^t(\delta_\gamma(B)).
 \end{cases} 
\label{34}\en
 Now, suppose for a moment that $\delta_\gamma(A)$ and $\delta_\gamma(B)$ are both linear in $A$, $B$ and $\1$, i.e. that $\delta_\gamma(A)=k_AA+k_BB+k_{\1}\1$ and  $\delta_\gamma(B)=k'_AA+k'_BB+k'_{\1}\1$, for some complex set of parameters $k_A, k_B,...$. Then the above system becomes
 $$
 \begin{cases}
 	& \frac{d}{dt}\gamma^t(A)= k_A\gamma^t(A)+k_B\gamma^t(B)+k_{\1}\gamma^t(\1)\\
 	&\frac{d}{dt}\gamma^t(B)=k'_A\gamma^t(A)+k'_B\gamma^t(B)+k'_{\1}\gamma^t(\1).
 \end{cases} 
 $$ 
 If, in particular, $k_{\1}=k'_{\1}=0$, this is a closed system of coupled differential equations for the two unknowns $\gamma^t(A)$ and $\gamma^t(B)$ which can be easily solved, recalling that $\gamma^0(A)=A$ and $\gamma^0(B)=B$. If $k_{\1}$ or $k'_{\1}$, or both, are non zero, we could still solve the system above once we know $\gamma^t(\1)$, since it does not depend on $A$ and $B$, and can be seen as a sort of external driving force arising because of the non Hermiticity of $H$. However, the situation changes completely if either $\delta_\gamma(A)$, or $\delta_\gamma(B)$, or both, are non-linear in $A$ and $B$, since in this case $\gamma^t(AB)\neq\gamma^t(A)\gamma^t(B)$,  $\gamma^t(A^2)\neq(\gamma^t(A))^2$, and so on. This means that, the system (\ref{34}) is, most probably, not closed and, in the attempt to close the system, we need to add differential equations for these nonlinear terms in $\delta_\gamma(A)$ and in $\delta_\gamma(B)$. However, there is no guarantee that it is possible to close the system, in this way. Moreover, even if we can find the solution for some of the operators describing the system, it is not true in general that we can deduce other relevant time evolutions. A very simple example of this situation is discussed in the Appendix. Here we only want to suggest that a natural way out exists, but it is surely not satisfying: one could assume that the Heisenberg-like dynamics, even when $H\neq H^\dagger$, is only driven by $H$ so that, rather than (\ref{25}) or (\ref{26}), one consider, as starting point, the { rule} $\alpha_H^t(X)=e^{iHt}Xe^{-iHt}$. This is, as a matter of fact, the {\em dual assumption} of saying that in the Schr\"odinger representation, the dynamics is given by $i\dot \Psi(t)=H\Psi(t)$, giving no role to $H^\dagger$. However, $\alpha_H^t$ does not preserve the adjoint: $\alpha_H^t(X^\dagger)\neq (\alpha_H^t(X))^\dagger$ and, possibly more important, (\ref{24}) does not hold, since  $\pin{\varphi_0}{\alpha_H^t(X)\psi_0}=\pin{e^{-iH^\dagger t}\varphi(0)}{Xe^{-iHt}\psi(0)}$, which suggests two different Schr\"odinger dynamics for the states of $\Sc$ depending on the fact that they appear in the first or in the second part of the scalar product. This is quite unphysical, of course.

 \section{The intertwining operator between $H$ and $H^\dagger$}\label{sect4}
 
 In Section \ref{sect3} we have only considered what happens in presence of an Hamiltonian $H$ which is different from its adjoint $H^\dagger$, neglecting completely the role of the operators $S_\varphi$ and $S_\Psi$ we have considered in Section \ref{sect2}, see in particular (\ref{23}). In this section we will focus on what the existence of this operator implies. This will be particularly useful in connection with symmetries and integrals of motion, as we will see soon. However, before going on, it is useful to recall once more that when $dim(\Hil)=\infty$, and in particular in presence of unbounded operators, the use of these operators is rather subtle, and requires deep results in functional analysis and in operator theory, \cite{bagbook}.
 
 To simplify the notation we will simply put $S=S_\Psi$. Hence $S_\varphi=S^{-1}$, and we get the following useful (and well known) relations:
 \be
 H^\dagger=SHS^{-1}, \qquad e^{iH^\dagger t}=Se^{iHt}S^{-1},
 \label{41}\en
 whose mathematical proofs are here strongly connected to the fact that all the operators involved are bounded\footnote{In particular, since $H\neq H^\dagger$, the existence itself of $e^{iHt}$ for unbounded $H$ is not guaranteed by the spectral theorem, neither it can be checked with a norm convergent sum.}.
 
 A first easy result concerns $\gamma^t(\1)$, which, as we have discussed before, plays a relevant role in the $\gamma$-dynamics. Using (\ref{41}) we can write
 \be
 \gamma^t(\1)=S\alpha_H^t(S^{-1})=(\alpha_H^t(S^{-1}))^\dagger S,
 \label{42}\en
 where we recall that $\alpha_H^t(X)=e^{iHt}Xe^{-iHt}$. The companion infinitesimal result of this formula is the following:
 \be
 \delta_\gamma(X)=S\delta_H(S^{-1}X),
 \label{43}\en
 where we have introduced, with obvious notation, $\delta_H(Y)=i[H,Y]$. Notice that $\delta_H$ is  not a *-derivation, since $\delta_H(Y^\dagger)\neq (\delta_H(Y))^\dagger$. Formula (\ref{43}) can be extended to powers of $\delta_\gamma$, and the extension can be proved by induction. In particular we have
   \be
  \delta_\gamma^l(X)=S\delta_H^l(S^{-1}X),
  \label{44}\en
 where $l=0,1,2,\ldots$. An immediate consequence of this equality is an alternative way to deduce what stated in Proposition \ref{prop2}. Indeed we have
 $$
 \sum_{k=0}^{\infty}\frac{t^k\delta_\gamma^k(X)}{k!}=S\sum_{k=0}^{\infty}\frac{t^k\delta_H^k(S^{-1}X)}{k!}=Se^{itH}(S^{-1}X)e^{-itH}=e^{itH^\dagger}Xe^{-itH}=\gamma^t(X).
 $$
 In this derivation we have also used the continuity of $S$, and we have re-summed the series adapting the same arguments of Proposition \ref{prop2} to $\alpha_H^t$.
 
 \begin{defn}\label{simmetria}
 	An operator $X\in\B(\Hil)$ is called a $\gamma$-symmetry if\be [H,S^{-1}X]=0.\label{45}\en
 \end{defn} 
 In the next Proposition we show, among other results, the relation between $\gamma$-symmetries and integrals of motion (or, maybe more correctly, $\gamma$-motion), and with intertwining operators.
 
 \begin{prop}\label{prop4}
 An operator $X\in\B(\Hil)$ is a $\gamma$-symmetry if and only if any of the following statements, all equivalent, are satisfied.
 
 \begin{enumerate}
 	\item \be [H^\dagger,X^\dagger S^{-1}]=0; \label{46}\en
 	
 	\item \be H^\dagger X=XH; \label{47}\en
 	
 	\item \be \delta_\gamma(X)=0; \label{48}\en
 	
 	\item \be \gamma^t(X)=X. \label{49}\en
 
 \end{enumerate}
 \end{prop}

\begin{proof}
	First we observe that (\ref{46}) is just the hermitian conjugate of (\ref{45}). (\ref{47}) is equivalent to (\ref{45}) because of (\ref{41}), while (\ref{48}) is equivalent to (\ref{47}) because of (\ref{32}). Finally, if $\delta_\gamma(X)=0$, then by (\ref{33}) $\frac{d}{dt}\gamma^t(X)=0$, and therefore $\gamma^t(X)=\gamma^0(X)=X$. Viceversa, if $\gamma^t(X)=X$, then $\frac{d}{dt}\gamma^t(X)=0$ so that, again by (\ref{33}), $\gamma^t(\delta_\gamma(X))=0$. Now, since both $e^{iH^\dagger t}$ and $e^{-iHt}$ are invertible, (\ref{48}) follows. 
		
\end{proof}
 
 A first obvious, and reasonable, consequence of these results is that $S$ is a $\gamma$-symmetry, see (\ref{45}), and, as such, intertwines between $H$ and $H^\dagger$, see (\ref{47}) and (\ref{23}), and does not evolve in time: $\gamma^t(S)=S$. It is interesting to observe that, on the other hand, $S^{-1}$ is not a $\gamma$-symmetry. This is because, in general, $[H,S^{-2}]\neq0$.  This fact is related, of course, to our definition of $\gamma^t$.
 
 \vspace{2mm}
 
 In {\em ordinary} quantum mechanics, the role of integrals of motion is given to those observables which commute with the Hamiltonian, if any. Here we see that this is not really so, since integrals of motion do not commute with $H$, while their composition in (\ref{45}) do. So it is interesting to check what is, in this context, the role of those operators which do commute with $H$. This is the content of the following lemma:
 
 \begin{lemma}
 	Let $X\in\B(\Hil)$ be a $\gamma$-symmetry, and let $Y\in\B(\Hil)$ a second operator commuting with $H$. Then $XY$ is a $\gamma$-symmetry.
 \end{lemma}

\begin{proof}
	Indeed we have
	$$
	[H,S^{-1}(XY)]=[H,S^{-1}X]Y+S^{-1}X[X,Y]=0,
	$$
	under our assumptions on $X$ and $Y$.

\end{proof}

This lemma has, in our opinion, an interesting interpretation: the operators commuting with $H$ are not  $\gamma$-symmetries by themselves, but can still be used to deform $\gamma$-symmetries into new $\gamma$-symmetries. On the other hand, if $X$ and $Y$ are both $\gamma$-symmetries, $XY$ is not, since in general $[H,S^{-1}(XY)]=S^{-1}X[H,S]S^{-1}Y\neq0$. Of course, if $[H,S]=0$, then $[H,S^{-1}(XY)]=0$ as well. But in this case the situation is not so interesting, since it is only compatible with the case in which $H=H^\dagger$, see (\ref{23}). On the other hand, in principle, if $X$ or $Y$, or both, are not invertible, $[H,S^{-1}(XY)]$ could be zero even when $[H,S]\neq0$. For instance, this happens if the range of the operator $(H-H^\dagger)Y$ belongs to the kernel of $X$.

Proposition  \ref{prop4} describes also the relation between $\gamma$-symmetries and intertwining operators (IOs). An operator $X$ is called and IO between $H$ and $H^\dagger$ if it satisfies (\ref{47}). We refer to \cite{bagint1,midya,lisok} for some results and applications of this kind of operators. It is easy to find solutions of (\ref{47}). Any operator of the kind
\be
X=\sum_{k=0}^Nx_k|\Psi_k\left>\right<\Psi_k|,
\label{410}\en
satisfies the equality in (\ref{47}), for all possible choices of $x_n\in\mathbb{C}$. In particular, if $x_n\neq0$ for all $n$, $X$ is invertible and $X^{-1}=\sum_{k=0}^Nx_k^{-1}|\varphi_k\left>\right<\varphi_k|$. Some easy comments are the following: (a) if $x_n\in\mathbb{R}$ for all $n$, then $X=X^\dagger$; (b) if $x_n=1$ for all $n$, then $X=S$; (c) $X\varphi_n=x_n\Psi_n$. Hence, $X$ maps eigenstates of $H$ into eigenstates of $H^\dagger$ with, in general, a different normalization. Of course, for this to be true, we must have $x_n\neq0$; (d) the fact that in (\ref{410}) we are considering finite sums could be relaxed, but it is still useful (here, at least) to assume that the series we get sending $N\rightarrow\infty$ converges in the norm of $\B(\Hil)$. In fact, under this assumption, $X\in\B(\Hil)$. However, in this case, even if $x_n\neq0$ for all $n\in\mathbb{N}$, we cannot be sure that $X^{-1}$ is bounded as well.

Summarizing, the metric operator $S$ is a $\gamma$-symmetry, but a $\gamma$-symmetry is something more general than the metric operator $S$. In particular, while $S$ must be positive, positivity is not required to $X$.

\vspace{2mm}

{\bf Remark:--} In what we have done so far, we have considered the case of non degenerate $H$: all the (real) eigenvalues have multiplicity one. It is a simple exercise to check that all the results we have deduced in this paper, including the existence and the analytic expression of the IOs, still hold true or can be easily adapted even when degenerate eigenvalues exist. On the other hand, losing reality of the eigenvalues produces serious  changes in some of the results. This is because, in particular, $H$ and $H^\dagger$ are no longer isospectral, and therefore they do not admit an IO.

\section{More adjoints}\label{sect5}

In the existing literature on non self-adjoint Hamiltonians the role of the adjoint $\dagger$ is not unique, see \cite{benbook,bagprocpa,specissue2021,bagspringer} and references therein. Other adjoint maps can be defined, in a rather natural way, due to the fact that different scalar products on the same Hilbert space can be considered and play a relevant role in the analysis of the system.

Recently, in \cite{baghat}, the construction of a chain of isospectral Hamiltonians (or isospectral in pairs, if some eigenvalue is complex) have been proposed using some special features of a given Hamiltonian $H\neq H^\dagger$, of the kind we are considering here. We devote this section to a brief analysis of $\gamma$-dynamics in the context considered in \cite{baghat}, which is however much more general than the situation we will consider here. In particular, rather than considering the (virtually infinite) hierarchy of operators constructed in \cite{baghat}, we will restrict to the first few. Moreover, we will change a little bit the notation coherently with the one adopted in this paper.

The starting point is the operator $S=S_\Psi$ which can be used, together with $S^{-1}$ to introduce two different scalar products
\be
\pin{f}{g}_\flat:=\pin{S^{-1}f}{g}, \qquad\qquad \pin{f}{g}_\sharp:=\pin{Sf}{g},
\label{51}\en
$\forall f,g\in\Hil$. The fact that these are really scalar product has been discussed in many papers. We refer to \cite{bagbookPT}, for instance. We can then introduce two adjoint maps as follows
\be
\pin{Xf}{g}_\flat=\pin{f}{X^{\flat}g}_\flat, \qquad\qquad \pin{Xf}{g}_\sharp=\pin{f}{X^{\sharp}g}_\sharp,
\label{52}\en
$\forall f,g\in\Hil$. It turns out that
\be
X^\flat=SX^\dagger S^{-1}, \qquad X^\sharp=S^{-1}X^\dagger S,
\label{53}\en
which satisfy the properties of the adjoint. For instance, both maps are linear and $(X^\flat)^\flat=(X^\sharp)^\sharp=X$. We observe also that
\be
(X^\flat)^\dagger=(X^\dagger)^\sharp,
\label{54}\en
which implies that the order of $\dagger$, $\sharp$ and $\flat$ is important.
Therefore, other than $H^\dagger$, we can also consider the operators $H^\flat$ and $H^\sharp$, and combinations of these adjoints (as, for instance, $(H^\flat)^\dagger$, $(H^\sharp)^\dagger$, etc...). However, due to (\ref{23}) and (\ref{53}), some of these maps are redundant. In fact, we can check that
\be
H=H^\sharp, \qquad\qquad H^\dagger=(H^\dagger)^\flat,
\label{55}\en
showing that, if all the $E_n$'s are real, $H$ is self-adjoint with respect to $\sharp$, and $H^\dagger$ is also self-adjoint, but with respect to $\flat$. Incidentally we observe that the two equalities in (\ref{55}) are equivalent, because of (\ref{54}). Summarizing, the introduction of the new scalar products in (\ref{51}) implies the existence of a new independent\footnote{Notice that, however, $H^\flat$ is similar to $H^\dagger$, because of (\ref{53})} Hamiltonian, $H^\flat$, which should be added to $H$ and $H^\dagger$. Following the idea in \cite{baghat}, it is easy to find the eigenstates of $H^\flat$, and to prove that it is isospectral to $H$ and $H^\dagger$. For that, we introduce the vectors $\xi_n=S\Psi_n$, $n=1,2,\ldots,N$. Hence $H^\flat\xi_n=SH^\dagger S^{-1}S\Psi_n=E_n\Psi_n$. Hence we have
\be
H\varphi_n=E_n\varphi_n, \qquad H^\dagger\Psi_n=E_n\Psi_n, \qquad H^\flat\xi_n=E_n\xi_n, \qquad (H^\flat)^\dagger\eta_n=E_n\eta_n,
\label{56}\en 
for all $n$. Here $\F_\eta=\{\eta_n\}$ is the set of vectors biorthonormal to $\F_\xi=\{\xi_n\}$,  defined by $\eta_n=S^{-1}\varphi_n$: $\pin{\eta_n}{\xi_m}=\pin{\varphi_n}{\Psi_m}=\delta_{n,m}$. In particular, the last eigenvalue equation above can be deduced easily: $(H^\flat)^\dagger\eta_n=S^{-1}HS S^{-1}\varphi_n=S^{-1}H\varphi_n= E_n S^{-1}\varphi_n=E_n\eta_n$. 

In \cite{baghat} the families $\F_\eta$ and $\F_\xi$ are used to construct new operators as we did in Section \ref{sect2} for $S_\varphi$ and $S_\Psi$, new scalar products and new adjoint maps. Here we only want to enrich Proposition \ref{prop4} by adding another equivalent condition to those listed there.

 \begin{prop}\label{prop5}
	An operator $X\in\B(\Hil)$ is a $\gamma$-symmetry if and only if 
	\be H^\flat (SX)=(SX) H.
	\label{57}\en
	\end{prop}
\begin{proof}
Indeed we have 
$$
H^\flat (SX)=SH^\dagger S^{-1}SX=SH^\dagger X=(SX)H,
$$
using (\ref{47}).

\end{proof}
This proposition is not surprising, since it is in agreement with the fact that $H$ and $H^\flat$ are isospectral. In this case the operator which intertwines between the two Hamiltonians is not $X$, as it was in (\ref{47}), but $SX$.

Of course, there is no non trivial (i.e., different from $\1$) IO between $H$ and $H^\sharp$ since they are, in fact, the same operator. Putting together the results of Propositions \ref{prop4} and \ref{prop5} we conclude, in particular, that if $X\in\B(\Hil)$ is such that $H^\flat (SX)=(SX) H$, then $\gamma^t(X)=X$. Hence we have still another method to deduce the integrals of motion for our system.

\section{Conclusions}\label{sectconcl}

This paper is another step in the direction of a full construction of an algebraic approach to the dynamics of quantum systems driven by non self-adjoint Hamiltonians. Similar ideas have been considered by several authors in many papers, but we believe this is the first attempt to create a coherent settings which can be successfully used, and possibly generalized to many other relevant situations, like to
 quantum systems living in infinite-dimensional Hilbert spaces. This extension, as already observed, is mathematically non-trivial, and will require a much more technical analysis than that considered here. Another non trivial extension is needed for those Hamiltonians which admit at least some non real eigenvalues and, even more interesting, for those in which some phase transition (from unbroken to broken phase) is observed, since these operators are particularly interesting for concrete applications. 

From a more physical point of view, it would be very interesting connect the general settings proposed in this paper with some concrete applications as those discussed, in particular, in \cite{yog1,sato}. This can contribute to enrich the already existing link between mathematics and physics of non self-adjoint Hamiltonians.

\section*{Acknowledgements}

The author acknowledges partial financial support from Palermo University (via FFR2021 "Bagarello") and from G.N.F.M. of the INdAM. The authors expresses his gratitude to Naomichi Hatano and Yogesh Joglekar for many interesting discussions at an early stage of this research.

\section*{Fundings}

The author acknowledges partial financial support from Palermo University (via FFR2021 "Bagarello").

\section*{Conflicts of interest}

There are no conflicts of interest.

\section*{Availability of data and material}

Not applicable.

\section*{Code availability}

Not applicable.

\section*{Authors' contributions}

Not applicable.

\renewcommand{\theequation}{A.\arabic{equation}}

\section*{Appendix: a simple nonlinear model}\label{appendix}

Let our system $\Sc$ be described in terms of two fermionic modes $a_1$ and $a_2$, with $\{a_i,a_j^\dagger\}=\delta_{i,j}\1$, $\{a_i,a_j\}=0$, $i,j=1,2$. Here $\1$ is the identity operator on the Hilbert space of $\Sc$, $\Hil=\mathbb{C}^4$. We assume that $\gamma^t$ is generated by $H=N_1a_2$, where $N_1=a_1^\dagger a_1$. The physical interpretation of $H$ is not so relevant, here. However, following what discussed in \cite{bagbook1,bagbook2}, we could imagine $H$ describing the lowering of the density of a second {\em biological} species in presence of a first one\footnote{More explicitly, the second species could be the food of the first one, and for this reason, in presence of the first species, the density of the second decreases.}.

Since $H^\dagger=N_1a_2^\dagger$ we find
\be
\begin{cases}
	& \frac{d}{dt}\gamma^t(a_1)=-i\gamma^t(a_1a_2) \\
	&\frac{d}{dt}\gamma^t(a_2)=i\gamma^t(N_1N_2),
\end{cases} 
\label{a1}\en
 where $N_2=a_2^\dagger a_2$. This simple system already shows what stated in Section \ref{sectnonauto}: the system is not closed. And this is not a consequence of the appearance of $a_1^\dagger$ and $a_2^\dagger$ in the second equation, but follows from the failure of the automorphism property for $\gamma^t$. For instance, we cannot expect that $\gamma^t(a_1a_2)$ and $\gamma^t(a_1)\gamma^t(a_2)$ coincide. And, indeed, this is what we will now check explicitly, for this simple model.
 
 The equation of motion for  $\gamma^t(a_1a_2)$ and $\gamma^t(N_1N_2)$ are particularly easy. In fact, from (\ref{33}), and noticing that $\delta_\gamma(a_1a_2)=i\left(H^\dagger a_1a_2-a_1a_2 H\right)=0$, we find $\frac{d}{dt}\gamma^t(a_1a_2)=0$, so that $\gamma^t(a_1a_2)=\gamma^0(a_1a_2)=a_1a_2$.
Analogously we find that $\delta_\gamma(N_1N_2)=0$, so that $\gamma^t(N_1N_2)=N_1N_2$. This implies that the system (\ref{a1}) can be easily solved, and the have
$\gamma^t(a_1)=-ia_1a_2t+a_1$ and $\gamma^t(a_2)=iN_1N_2t+a_2$. It is now clear that, as we have anticipated,  $\gamma^t(a_1a_2)\neq \gamma^t(a_1)\gamma^t(a_2)$. 

Notice that, for this particular model, the knowledge of $\gamma^t(a_j)$ implies also the knowledge of $\gamma^t(a_j^\dagger)$. However, to compute $\gamma^t(N_j)$,  it is not correct to compute $\gamma^t(a_j^\dagger)\,\gamma^t(a_j)$, since this is, in general, different from $\gamma^t(a_j^\dagger a_j)$. Hence, to compute $\gamma^t(N_j)$, we need to use (\ref{33}) directly on $N_j$.

\end{document}